\documentclass[12pt]{iopart}

\usepackage{dcolumn}
\usepackage{bm}
\usepackage{verbatim}       

\usepackage[dvips]{graphicx}
\usepackage{amsthm}
\usepackage{amssymb}
\usepackage{bm}
\usepackage{amstext}
\usepackage{psfrag}

\expandafter\let\csname equation*\endcsname\relax

\expandafter\let\csname endequation*\endcsname\relax

\usepackage{amsmath}
\usepackage{amsfonts}
\usepackage{mathrsfs}
\usepackage{cite}

\newcommand{\dd}{{\rm d}}

\newcommand{\bd}{\begin{definition}}                
\newcommand{\ed}{\end{definition}}                  
\newcommand{\bc}{\begin{corollary}}                 
\newcommand{\ec}{\end{corollary}}                   
\newcommand{\bl}{\begin{lemma}}                     
\newcommand{\el}{\end{lemma}}                       
\newcommand{\bp}{\begin{proposition}}            
\newcommand{\ep}{\end{proposition}}                
\newcommand{\bere}{\begin{remark}}                  
\newcommand{\ere}{\end{remark}}                     

\newcommand{\bt}{\begin{theorem}}
\newcommand{\et}{\end{theorem}}

\newcommand{\bit}{\begin{itemize}}
\newcommand{\eit}{\end{itemize}}
\newtheorem{theorem}{Theorem}[section]
\newtheorem{corollary}[theorem]{Corollary}
\newtheorem{lemma}[theorem]{Lemma}
\newtheorem{proposition}[theorem]{Proposition}
\theoremstyle{definition}
\newtheorem{definition}[theorem]{Definition}
\theoremstyle{remark}
\newtheorem{remark}[theorem]{Remark}

\begin{document}

\title{On the existence of smooth Cauchy steep time functions}


\author{E Minguzzi}
\address{Dipartimento di Matematica e Informatica ``U. Dini'', Universit\`a
degli Studi di Firenze, Via S. Marta 3,  I-50139 Firenze, Italy.}

\ead{ettore.minguzzi@unifi.it}


\date{}

\begin{abstract}
\noindent A simple proof is given that every globally hyperbolic spacetime admits a  smooth Cauchy steep time function. This result is useful in order to show that globally hyperbolic spacetimes can be isometrically embedded in Minkowski spacetimes and that they spit as a product. The proof is based on a recent result on the differentiability of Geroch's volume functions.
\end{abstract}

\section{Introduction}
Let $(M,g)$ be a spacetime, namely a Hausdorff, paracompact,
connected, time-oriented Lorentzian manifold, endowed with a ($C^3$)
metric $g$ of signature $(-,+,$ $\cdots,+)$. Since every $C^r$
manifold is $C^r$ diffeomorphic to a $C^\infty$ manifold
\cite{hirsch76}, we shall not bother with the  the degree of
differentiability of $M$, as we can suppose to have been given a
smooth atlas.

On a spacetime a {\em time function} $\tau:M\to \mathbb{R}$ is a
continuous function which increases over every causal curve, that
is, $x<y \Rightarrow f(x)<f(y)$, where as usual, we write $x<y$ if
there is a future-directed causal curve connecting $x$ to $y$. A
{\em Cauchy time function} is a time function for which the level
sets are Cauchy hypersurfaces, namely closed acausal  sets
intersected exactly once by any inextendible causal curve. A time
function $\tau$ is {\em steep} if it is continuously differentiable
and $-g(\nabla\tau,\nabla \tau)\ge 1$.

The existence of smooth Cauchy time functions in globally hyperbolic
spacetimes has been proved through different approaches in
\cite{bernal03,fathi12,chrusciel13}. It implies  that these
spacetimes admit a smooth splitting as a product $\mathbb{R}\times
S$, $g=-\beta^2\dd t^2\oplus h_t$, $\beta:M\to (0,+\infty)$, where
$t$ is the Cauchy time function and $h_t$ is a time dependent
Riemannian metric over the Cauchy hypersurface $\{t\}\times S$.
Recently, M\"uller and S\'anchez \cite{muller11} proved that the
Cauchy time function can be found steep, a fact which implies that
$(M,g)$ can be embedded in $N$-dimensional Minkowski spacetime for
some $N\ge 2$, and that $\beta$ can be chosen such that $\beta\le
1$. This result can also be used to express a formula for the Lorentzian distance in terms of the family of time functions \cite{rennie14}.

In this work we provide a simple proof of  the existence of smooth
Cauchy steep time functions in globally hyperbolic spacetimes by
using some recent results on Geroch's volume functions.


Let $(M,g)$ be a spacetime and let us consider a non-negative
continuous function $\varphi$. The Geroch's volume function is
\begin{equation}
 \tau^\pm_\varphi(p) = \int_{J^{\pm}(p)} \varphi  \, d\mu_g
 \;,
\end{equation}
where $d\mu_g$ is the volume element of $g$.

In a recent joint work with J.\ Grant and  P.\
Chru\'sciel \cite{chrusciel13} we proved the following lemma, which basically
follows from the differentiability properties of the
exponential map.

\begin{lemma}
 \label{nus}
In globally hyperbolic spacetimes the functions $\tau^\pm_\varphi$
are continuously differentiable for all continuous compactly
supported non-negative functions $\varphi$. The gradient at $p\in M$
is past-directed timelike or vanishing depending respectively on whether
$E^{\pm}(p)$ intersects or not the set $\{q: \varphi(q)>0\}$.
\end{lemma}
For completeness we mention that the gradient is
given by
\[
\nabla_X \tau^\pm_{\varphi}=\int_{E^{\pm}(p)} \varphi  \, J(X) \rfloor d\mu_g
 \;,
 \]
where $J(X)$ is the Jacobi field obtained solving the Jacobi
equation over each  generator $\gamma(s)$ of the horismos
$E^{\pm}(p)$, with initial condition $J(0)=X$, $(\frac{d}{d s}
J)(0)=0$.

In the  proof of the main  Theorem \ref{ciz} it will be useful to keep in
mind the next topological result.

\begin{lemma} \label{nif}
Let $K$ be a compact subset of a topological space and let
$\{U_s\}_{s\in S}$ be a locally finite family of subsets. Then $K$
is intersected  by at most finitely many elements of the family.
\end{lemma}

\begin{proof}
Let $A_i$ be the open (in $K$) subset of $K$ made of points which
admit a neighborhood which intersects at most  $i$ elements of
$\{U_s\}_{s\in S}$. Clearly, $A_i\subset A_{i+1}$, and $\bigcup_i
A_i=K$. However, this open covering of $K$ admits a finite
subcovering which proves that $K=A_j$ for some $j$. Since every
$p\in K$ admits an open neighborhood $N_p$ which intersect at most
$j$ elements of the family $\{U_s\}_{s\in S}$, and since $K$ admits
a finite subcovering of say,  $r$ elements of the form $N_{p_i}$,
$K$ is intersected at most by $rj$ elements of the family
$\{U_s\}_{s\in S}$.
\end{proof}

We are ready to prove the existence of smooth Cauchy steep time
functions. The proof uses just the previous Lemmas and Geroch's topological splitting theorem.

\begin{theorem} \label{ciz}
Let $(M,g)$ be a globally  hyperbolic spacetime. There exists a smooth
 Cauchy time function $\tau\colon M\to \mathbb{R}$
with timelike past-directed gradient $\nabla \tau$ which is steep,
namely  $-g(\nabla \tau,\nabla \tau)\ge 1$.
\end{theorem}

\begin{proof}
Just in this proof we shall say that $\tau$ is `steep' if $-g(\nabla
\tau,\nabla \tau)> 1$ with the strict inequality.  According to
Geroch's topological splitting \cite{geroch70,hawking73} there is a
continuous Cauchy time function $t\colon M\to \mathbb{R}$. The level sets
$S_t$ are Cauchy hypersurfaces. We are first going to construct a
continuously differentiable function $\tau^{-}$ (resp.\ $-\tau^+$) over $M$,
with past-directed timelike or  vanishing gradient, which is steep
over $J^{+}(S_0)$ (resp.\  $J^{-}(S_0)$) and such that $\tau^--t>0$
over $J^+(S_0)$ (resp. $-\tau^+<t$ over $J^{-}(S_0)$) . Then
$\tau'=\tau^{-}-\tau^+$ will be clearly continuously differentiable,
steep and Cauchy over $M$.

Next, by \cite[Theor.\ 2.6]{hirsch76} $C^\infty(M,\mathbb{R})$ is dense in
$C^1(M,\mathbb{R})$ endowed with the Whitney strong topology
\cite[p.\ 35]{hirsch76}, thus we  can find $\tau\in
C^\infty(M,\mathbb{R})$ which approximates $\tau'$, up to the first
derivative, as accurately as we want over $M$. In particular, we can
find $\tau$ such that $\vert \tau-\tau'\vert<1$ and $-g(\nabla
\tau,\nabla \tau)> 1$, where the former inequality implies that
$\tau$ is Cauchy, and the latter inequality implies that $\tau$ is
steep.

Let us construct $\tau^-$, the plus case being analogous. For every
$p\in J^+(S_0)$ there is a open neighborhood $V_p$ with compact
closure  contained in $I^+(S_{-1})$. Let $V_0=M\backslash
J^{+}(S_0)$, by paracompactness the open covering $\{V_0\}\cup\{
V_{p}: p\in J^+(S_0)\}$ admits a locally finite refinement
(necessarily countable by Lemma \ref{nif} and $\sigma$-compactness
of $M$) and a corresponding partition of unity $\{\varphi_i, i\ge
0\}$, where, defined $U_j=\{q:\varphi_j(q)>0\}$ we have for $j\ge 1$,
$\textrm{supp} \, \varphi_j=\overline{U_j}\subset I^+(S_{-1})$.



Let $K_i\subset J^{+}(S_0)$, be a sequence of compact sets such that $K_0=\emptyset$, $\bigcup_i
K_i=J^{+}(S_0)$, $K_i=J^{-}(K_i)\cap J^{+}(S_0)$ and
\[
K_i\cup \bigcup_{j: U_j \cap K_i\ne \emptyset} \overline{U_j}
\subset K_{i+1}.
\]
(The indices $j$ entering the union are finite in number due to
Lemma \ref{nif}.) Let $\lambda_j$ be a sequence of positive numbers
and let $\varphi=\sum_{j=1}^\infty \lambda_j \varphi_j$. Since
$\{\textrm{supp} \, \varphi_j\}$ is locally finite, $\varphi$ is
finite as at any point only a finite number of terms give a
non-vanishing contribution. Moreover, $\textrm{supp} \, \varphi
=\cup_{j=1} \overline{U_j}\subset I^+(S_{-1})$ (it is useful to
recall that the closure operator is additive if the union is over a
locally finite family \cite[Lemma 20.5]{willard70},
 and that the closure of a locally finite
family is locally finite).

We are going to find a sequence $\lambda_j$ such that $
\tau^-:=\tau^{-}_{{\varphi}}$ is steep over $J^{+}(S_0)$. Observe
that
\[
\tau^{-}_{{\varphi}}=\sum_j \lambda_j \tau^{-}_{{\varphi_j}}, \qquad
\nabla \tau^{-}_{{\varphi}}=\sum_j \lambda_j \nabla
\tau^{-}_{{\varphi_j}}.
\]
By induction suppose that we can find a finite sequence
$\lambda_1^i, \lambda_2^i, \cdots, \lambda_{n_i}^i$,
($\lambda^i_{j}=0$ for $j>n_i$) such that $\tau^-$, defined as
above, is steep over $K_i$, and satisfies the inequality
$\tau^--t>0$ over $K_i$. Let $\Lambda_i$ be the finite index set
such that if $r\in \Lambda_i$ then $U_r\cap (K_{i+1}\backslash
\textrm{Int}K_i) \ne \emptyset$. We have $U_r\cap
K_{i-1}=\emptyset$, for otherwise $U_r\subset K_i \Rightarrow U_r
\subset \textrm{Int} K_i$, a contradiction. Let $p\in
K_{i+1}\backslash \textrm{Int} K_i$, then $p$ belongs to some $U_r$,
$r \in \Lambda_i$, and hence   $E^{-}(p)$ intersects $U_r$.

We now replace the constants $\lambda_r$, $r\in \Lambda_i$,
$\lambda_r \to \lambda_r'$, with larger constants chosen in such a
way that $\tau^-$ and $\nabla \tau^-$ get replaced by
\[
\tau^-_{{\varphi}}\to \tau^-_{{\varphi}}+ \sum_{r\in
\Lambda_i}(\lambda_r'-\lambda_r) \tau^-_{\varphi_r}, \qquad \nabla
\tau^-_{{\varphi}}\to \nabla\tau^-_{{\varphi}}+ \sum_{r\in
\Lambda_i}(\lambda_r'-\lambda_r) \nabla \tau^-_{\varphi_r}.
\]
and so that $\tau^-_{{\varphi}}$ becomes steep at $p$, and
$\tau^-_{{\varphi}}-t$ becomes positive at $p$.


Since $K_{i+1}\backslash \textrm{Int} K_i$ is compact and $\tau^--t$
and $\nabla \tau^-$  depend continuously on each $\lambda_j$, the
redefinition of $\lambda_r$ and consequently of $\tau^-$ can be done
so as to obtain steepness and positivity of $\tau^--t$ all over
$K_{i+1}$. Observe that this redefinition can only increase the
value of $\tau^-$ over $K_i$ and that it certainly does not spoil
steepness there. Furthermore, it does not change $\tau^-$ over
$K_{i-1}$ thus the built inductive process leads to the desired
function $\tau^-$.
\end{proof}

\section*{Acknowledgments}
I thank James Grant and Piotr Chru\'sciel for useful and motivating
discussions on this and related problems.  
This work has been partially supported by GNFM of INDAM.\\


\end{document}